\documentclass[notitlepage]{revtex4-1}

\usepackage{amsmath,amssymb,amsthm,amsfonts}
\usepackage[colorlinks]{hyperref}
\hypersetup{
 pdftitle = {},
 pdfauthor = {},
 pdfcreator = {\LaTeX\ with package \flqq hyperref\frqq}
}

\usepackage{graphicx}
\usepackage{hyperref}

\newcommand{\be}{\begin{equation}}    
\newcommand{\ee}{\end{equation}}    

\newtheorem{dfn}{Definition}[section]
\newtheorem{obs}{Remark}[section]
\newtheorem{thm}{Theorem}[section]
\newtheorem{cor}{Corollary}[section]

\newcommand{\ec}{=\mathrel{\mathop:}}    
\newcommand{\ce}{\mathrel{\mathop:}=}  
\newcommand{\abs}[1]{\left\lvert#1\right\rvert}  

\newcommand{\om}{\kern.5em\overline{\phantom{M}}\kern-1.5em\mathcal{M}}

\newcommand{\vd}{\text{\rm vd}}
\newcommand{\vir}{\text{\rm vir}}

\DeclareMathOperator{\ev}{ev}   

\begin{document}


\title{Topological String Partition Function on Generalised Conifolds}


\author{Elizabeth Gasparim}
\email{etgasparim@gmail.com}
\affiliation{Departamento de Matem\'aticas, Universidad Cat\'olica del Norte, Antofagasta, Chile}

\author{Bruno Suzuki}
\email{obrunosuzuki@gmail.com}
\affiliation{Departamento de Matem\'aticas, Universidad Cat\'olica del Norte, Antofagasta, Chile}

\author{Alexander Torres-Gomez}
\email{alexander.torres.gomez@gmail.com}
\affiliation{Departamento de Matem\'aticas, Universidad Cat\'olica del Norte, Antofagasta, Chile, \, and\\
School of Physics \& Astronomy and Center for Theoretical Physics
Seoul National University, Seoul, Korea, \, and\\
Gauge, Gravity \& Strings, Center for Theoretical Physics of the Universe
Institute for Basic Sciences, Daejeon, Korea}

\author{Carlos A. B. Varea}
\email{carlosbassanivarea@gmail.com}
\affiliation{Instituto de Matem\'atica, Estat\'istica e Computa\c{c}\~{a}o Cient\'ifica, Universidade Estaudal de Campinas, Campinas, Brasil}


\date{\today}

\begin{abstract}
\vspace*{1cm}
We show that the partition function on a generalised conifold $C_{m,n}$ with ${m+n \choose m}$ crepant resolutions can be equivalently computed on the compound du Val singularity  $A_{m+n-1}\times \mathbb C$  with a unique crepant resolution. 
\end{abstract}

\pacs{}

\maketitle


\tableofcontents

\newpage

\section{Introduction}
The topological string partition function with target space on toric singular varieties  was defined in \cite{GKMR} via products of the partition functions on their crepant resolutions.  We could think of the crepant resolutions as all the possible evolutions of a singular toric variety where we went through a topological-changing transition from a singular to a non-singular variety. These resolutions are birationally equivalent but topologically distinct. Before choosing a specific resolution we could say that the variety is in a superposition of the different crepant resolutions. Thus, the singular toric variety free energy is defined as the sum of the contribution for the different topologies, crepant resolutions, and the partition function (that is, the exponential of the free energy) as the product. \\

One of the main motivations behind the work \cite{GKMR} was to check if in the computation of the topological string partition function with target space a generalised conifold, $C_{m,n}:=\{(x,y,z,w)\ | \ xy-z^m w^n=0\}\subset \mathbb{C}^4$, there exists a preferred crepant resolution. In other words, from a probabilistic point of view, if one of the resolutions carries  more importance (weight) than the others on the partition function. What they found out was that there was not such a preferred resolution or, in other words, that each resolution on the superposition of crepant resolutions has the same weight.\\

In this paper we go one step further and claim that although there is no a preferred resolution of $C_{m,n}$ it is possible to define a new conifold $C_{0,m+n}$ with a unique resolution that produces the same topological string partition function. Explicitly, the topological string partition function with target on the conifold $C_{m,n}$ is proportional to the one computed on $C_{0,m+n}$ to some power $d$, see below. Now, the variety $C_{0,m+n}\cong \mathbb C \times A_{m+n-1}$ is the compound du Val singularity $cA_{m+n-1}$, where $A_{m+n-1}=\{(x,y,w)\ | \ xy-w^{m+n}=0\}\subset \mathbb{C}^3$ is the $A_{m+n-1}$ singularity (orbifold).  Note that $cA_{m+n-1}$ has 1-dimensional singularities.\\
 
Thus, in a sense, we found a relation or duality between the topological string partition function with target space on the six dimensional generalised conifold $C_{m,n}$ and the one on the four dimensional $A_{m+n-1}$ orbifold.\\

In general, when we say that there is a duality we mean that there are two equivalent but different descriptions of the same phenomenon, in the sense that observables in both descriptions can be identified. Dualities are ubiquitous in theoretical physics and in particular in string theory. Their importance comes from the fact that when a duality exists, a system which looks extremely difficult to analyse using the current formulation becomes easier when the dual formulation is used. Some examples of dualities are T-duality, AdS/CFT holography, mirror symmetry, AGT correspondence, etc.\\

The kind of relation or duality that we explain in this paper is in the same spirit as T-duality (on the bosonic closed string theory) where two geometric background can be equivalently used to describe an observable as the partition function. Thus, although the dual varieties $C_{m,n}$ and $cA_{m+n-1}$ are geometrically and algebraically distinct, they are nevertheless physically equivalent as target spaces for topological strings on singular varieties.

\subsection{Topological string partition function}
Topological field theories of  Witten type (or cohomological type) have been used since they were discovered to study moduli spaces from a field theoretical point of view. They borrow tools from quantum field theory, which is the framework used to describe the standard model of particle physics, to answer pure mathematical questions.\\

Topological string theory studies maps from a source Riemann surface of genus $g$ with $n$  marked points to a target Calabi--Yau space. There are two types of topological strings known as type A and type B. The type A is about holomorphic maps and the type B about constant maps. The two types can be related by mirror symmetry, that might also  be called a mirror duality. In these article we are concerned with topological string theory of type A. \\

The topological string partition function has been used to study different invariants related to Calabi--Yau target spaces such as the Gromov--Witten invariants which count holomorphic algebraic curves on a Calabi--Yau threefold. Other kinds of invariants that can be described are Gopakumar--Vafa and Donaldson--Thomas. Although, Gopakumar--Vafa and Donaldson--Thomas invariants are not associated to maps between a source and a target spaces it has been shown that they are related to the topological string partition function. In particular, in the toric case it was shown in \cite{MNOP1}, \cite{MNOP2} that Donaldson--Thomas invariants are equivalent to Gromov--Witten invariants.\\

The type A topological string partition function $Z$ is given by the exponential of a generating function $F$, known as free energy in the physics literature, that is, $ Z=\exp(F)$. Now, the generating function $F$ is expressed as sum over the genus $g$ as \cite{CW}
\be
F=\sum_{g\geq 0} \lambda^{2g-2} F_g=\frac{1}{6 \lambda^2} a_{ijk} t^i t^j t^k- \frac{1}{24} b_i t^i+ F_\text{inst} \ ,
\ee
where $\lambda$ is the topological string coupling constant, $t^i$ are K\"ahler parameters of the Calabi--Yau, $a_{ijk}$ are triple intersection numbers, $b_i$ is related to the second Chern class of the Calabi--Yau and $F_\text{inst}$ has a non-polynomial dependence on $t^i$. The term $F_\text{inst}$ can be described by different formal series with coefficients given by the Gromov--Witten invariants or Gopakumar--Vafa invariants, Donaldson--Thomas invariants,
or Pandharipande--Thomas invariants,  which we denote as $F_\text{GW}$ or $F_\text{GV}$ or $F_\text{DT}$, or
$F_\text{PT}$,  respectively \cite{CW}.\\

\paragraph{Gromov--Witten invariants} The Gromov--Witten generating function is

\be
F_\text{GW}(\lambda, Q)=\sum_{g\geq 0} \sum_{\beta \in H_2^+(X)} N_{g,\beta}\ \lambda^{2 g-2} Q^\beta \ , 
\ee
where we have omitted the constants maps. The set of rational numbers $N_{g,\beta}$ are the Gromov--Witten invariants which count the number of holomorphic maps from a genus $g$ Riemann surface whose image is in the second homology class $\beta$ of the Calabi--Yau.  \\

We now recall the mathematical definition of the Gromov--Witten invariants. 
We write $\om_{g,n}(X, \beta)$ for the collection of maps from stable,
$n$-pointed curves of genus $g$ into $X$ for which
\be
  f_*[C] = \beta \in H_2(X;\;\mathbb{Z}) \text{ .} 
 \ee
$\om_{g,n}(X, \beta)$ has a virtual fundamental
class of virtual dimension
\be
\vd = (1-g) (\dim X - 3) - K_X(\beta) + n \text{ .} 
\ee
Consequently,
dimension of the classes $[\om_{g,n}(X,\beta)]^\vir$ is independent of
$\beta$ when $X$ is Calabi--Yau. Moreover, the
unpointed moduli $\om_{g,0}(X,\beta)$ has virtual dimension zero for
all $g$ if $\dim X = 3$, so on a three-dimensional Calabi--Yau,
$\om_{g,0}(X,\beta)$ ``counts" curves.
Here we recall the definition for the genus zero case, and refer the  reader to \cite[\S\,2]{MNOP2} for higher genera.

\begin{dfn}
Assume that the genus $g$ of the curve $C$ is zero, $g(C) = 0$. Let
\be
\ev_i \colon \om_{0,n}(X,\beta) \to X \text{ ,}\quad
   \bigl( f \colon (C; {P}) \to X \bigr) \mapsto f(P_i) 
\ee
be the $i^\text{th}$ evaluation map. Assume that $\sum_{i=1}^n
\deg(\gamma_i) = \vd$ for some $\gamma_i \in H^*(\om_{0,n}(X,
\beta))$. Then the \emph{genus-$0$ Gromov--Witten invariants} are
\be
\langle \gamma\rangle_\beta \ce \ev_1^*(\gamma_1) \cup \dotsb \cup
   \ev_n^*(\gamma_n) \cap [\om_{0,n}(X,\beta)]^\vir \text{ .}
\ee

\end{dfn}

When $\dim X = 3$, $X$ is Calabi--Yau (i.e.\ $K_X=0)$, arbitrary genus
$g$ and $n=0$, we have the \emph{unmarked Gromov--Witten invariants}
\[ N_{g,\beta}(X) \ce \int_{[\kern.3em\overline{\phantom{N}}\kern-.85em\mathcal{M}_{g,0}(X,\beta)]^\vir} 1 \text{ .} \]\\

\paragraph{Gopakumar--Vafa invariants} Counting BPS states on M-theory compactified on a Calabi--Yau treefold we obtain the Gopakumar--Vafa generating function as

\be\label{GV}
F_\text{GV}(\lambda,  Q)=\sum_{h\geq 0} \sum_{\beta \in H_2^+(X)} \sum_{d\geq 1} n_{h,\beta}\frac1 d  \left[ 2 \sin \left( \frac{d \lambda}{2} \right)  \right]^{2 g-2} Q^{d\beta}
\ee
where the   numbers $n_{h,\beta}$ appearing on the right hand side of (\ref{GV}) are called the Gopakumar--Vafa invariants. Although there is no intrinsic definition of the Gopakumar--Vafa invariants they can be defined recursively in terms of Gromov--Witten invariants, nevertheless a conjectural explicit description of the invariants is presented in \cite{MT}. These are conjectured to be integers, and 
from the string theory point of view, they count BPS states corresponding to M2-branes wrapping two cycles of homology class $\beta$ inside the Calabi--Yau.  \\

The GW partition function can be rewritten as \cite{CW}
\be
Z_\text{GW}= \left[ \prod_{k=1}^\infty (1- e^{\pm \lambda k})^{k \cdot n_{0,0}} \right] \left[ \prod_{k\geq1, \beta >0} (1-e^{\pm \lambda k} Q^\beta)^{k \cdot n_{0,\beta}} \right] \left[ \prod_{k\geq1, \beta >0} \prod_{l=0}^{2\beta-2} (1-e^{\pm \lambda (h-l-1)} Q^\beta)^{(-1)^{h+l}\cdot {2h-2 \choose l}\cdot n_{h,\beta}} \right].
\ee\\

\paragraph{Donaldson--Thomas invariants} Here we work with threefolds without compact 4-cycles. Even though Donaldson-Thomas invariants are defined for more general threefolds, we worked only  in the cases when there are no 4-cycles since this greatly simplifies the calculations.
Then, from the physical perspective, the Donaldson--Thomas invariants count the number of D6-D2-D0 BPS bound states, resulting from D6-branes wrapped on the Calabi--Yau threefold and D2-brane wrapped on a second homology class cycle $\beta \in H_2(X, \mathbb Z)$.\\

The Donaldson--Thomas representation of the partition function $Z_\text{DT}=\exp(F_\text{DT})$ is
\be \label{Z-DT}
Z_{\text{DT}}(X;q,v) = \sum_{\beta} \sum_{n\in\mathbb{Z}} \widetilde{N}_{n,
   \beta}(X) q^n v^\beta  \ ,
\ee
where the set of integers $ \widetilde{N}_{n,\beta}(X)$ are the Donaldson--Thomas invariants.\\

We now recall the mathematical definition of the Donaldson--Thomas invariants. As usual $\mathcal O_X$ denotes the structure sheaf of an algebraic variety $X$. An \emph{ideal subsheaf} $\mathcal{I}$ of $\mathcal{O}_X$ 
is 
a torsion-free rank-$1$ sheaf with trivial determinant. It follows that $\mathcal{I}^{\vee\vee}\cong \mathcal{O}_X$. 
Thus the evaluation map determines a quotient
\be \label{eq.ideal}
  0 \longrightarrow \mathcal{I} \xrightarrow{\ \ev\ } \mathcal{I}^{\vee\vee}
  \cong \mathcal{O}_X \longrightarrow \mathcal{O}_X\bigl/\mathcal{IO}_X = \imath_*\mathcal{O}_Y
  \longrightarrow 0 \text{ ,}
\ee
where $Y \subseteq X$ is the support of the quotient and
$\mathcal{O}_Y \ce (\mathcal{O}_X\bigl/\mathcal{IO}_X)\rvert_Y$ is the
structure sheaf of the corresponding subspace. Let $[Y] \in H_2(X;\;\mathbb{Z})$
denote the cycle 
determined by the $1$-dimensional components of $Y$.
We denote by
\be
I_n(X,\beta) 
\ee
the Hilbert scheme of ideal sheaves $\mathcal{I} \subset \mathcal{O}_X$
for which the quotient $Y$ in \eqref{eq.ideal} has dimension at most
$1$, $\chi(\mathcal{O}_Y) = n$ and $[Y] = \beta \in H_2(X; \;
\mathbb{Z})$.
If $X$ is a smooth, projective Calabi-Yau threefold,
then the virtual dimension of $  I_n(X,\beta)$  is zero, and we write
\be
\widetilde{N}_{n,\beta}(X) \ce \int_{[I_n(X,\beta)]^\vir} 1 
\ee
for the number of such ideal sheaves. These numbers can be assembled into a partition function, as in \eqref{Z-DT},
\be
 Z_{\text{DT}}(X;q,v) = \sum_{\beta} \sum_{n\in\mathbb{Z}} \widetilde{N}_{n,
   \beta}(X) q^n v^\beta = \sum_\beta Z_{\text{DT}}(X;q)_\beta v^\beta \text{ .} 
\ee
The degree-$0$ term is
\be
Z_{\text{DT}}(X;q)_0 = \sum_{n\geq0} \widetilde{N}_{n, 0}(X) q^n \text{.}
\ee
In the case of a  toric CY threefold $S$ such as  local curves or a local surfaces, 
one considers a toric compactification $X$ of $S$, and  then define 
\be
Z'_{\text{DT}}(S;q)_\beta =Z_{\text{DT}}(X;q)_\beta/Z_{\text{DT}}(X;q)_0\text{,}
\ee
see \cite[Sec.\thinspace 4.1]{MNOP1}. The invariants are then computed by localisation, 
using  weights corresponding to vertices and edges whose support does not intersect $X\setminus S$.\\

\paragraph{Pandharipande--Thomas invariants} Another curve counting invariant is defined in \cite{PT} 
using stable pairs, their definition is presented for projective varieties, but 
we find that it can be applied without troubles to the quasi-projective threefolds we consider, as our 
counting is determined on small neighbourhoods of curves, namely the exceptional sets of the CY resolutions. 
Given a threefold $X$,  consider the moduli space $P(X)$ of stable pairs $(F, s)$
where $F$ is a sheaf of fixed Hilbert polynomial supported in dimension
1 and $s \in H^0(X, F)$ is a section. The required  stability conditions are:
(i) the sheaf $F$ is pure, and
(ii) the section $s\colon \mathcal O_X \rightarrow  F$
 has 0-dimensional cokernel.
 Here purity  means every nonzero subsheaf of $F$ has support
of dimension 1. Let $P_n(X,\beta)$ denote the moduli space of semistable pairs $(s,F)$
where $F$ is a pure sheaf of Hilbert polynomial 
$$\chi(F(k))= k \int_\beta c_1(F) +n.$$
Given a class $\beta \in H_ 2(X, \mathbb Z)$, the stable pair invariants $P_{n, \beta}$   are by definition
the degree of the virtual cycle 
$$P_{n, \beta} = \int_{P_n(X,\beta)}1.$$
The partition function for the Pandharipande--Thomas theory is 
$$Z_{\text{PT}}(X;q) = \sum_n P_{n, \beta} \,  q^n.$$

The DT/PT correspondence proved in \cite{Br,To} gives 
$$Z_{\text{PT}}(X;q) = \frac{Z_{\text{DT}}(X;q)}{Z_{\text{DT}}(X;q)_0}.$$

\subsection{Summary of results}

The main result of this paper is that, given a generalised conifold $C_{m,n}:=\{(x,y,z,w)\ | \ xy-z^mw^n=0\}\subset \mathbb{C}^4$, the topological string partition function on $C_{m,n}$ is proportional to the one on the compound du Val singularity $cA_{m+n-1}$ to the power of $d$ (with $d$ computed from $m$ and $n$, see Theorem \ref{thm1}), that is,
\be
Z_\text{tot}(C_{m,n}) \sim (Z(cA_{m+n-1}))^d \ ,
\ee
where $Z_\text{tot}(C_{m,n})$ is defined as the product of partition functions of its crepant resolutions, 
\be
Z_\text{tot}(C_{m,n})=\prod_{i=1}^{N_\Delta} \, Z(C^{i}_{m,n}) \ .
\ee
Note that  $C_{m,n}$ has $N_\Delta={ m+n \choose m}$ crepant  resolutions and $cA_{m+n-1}$ has only one \cite{GKMR}. Thus, we have obtained a method to simplify considerably the computation of the topological string partition function on generalised conifolds, instead of computing ${ m+n \choose m}$ partition functions it is only necessary to compute one.\\

The procedure to find this result is described schematically on Figure \ref{diag}.\\

\begin{figure}
\begin{center}
\includegraphics[width=0.85\textwidth]{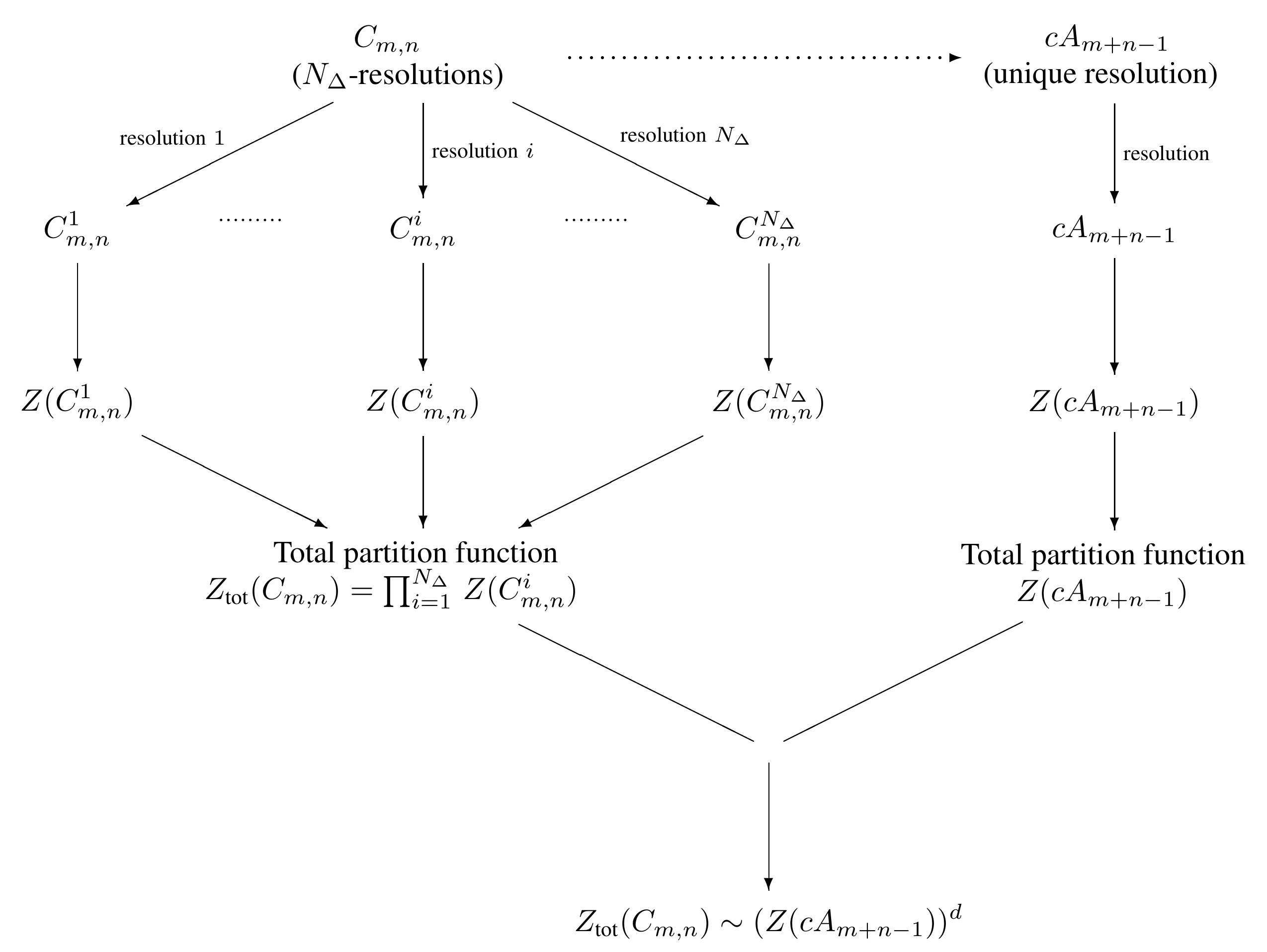}
\caption{Schematic representation of the steps followed to compute the partition function on the dual conifolds}
\label{diag}
\end{center}
\end{figure}

Apart from the computational advantage, we would like to stress that there is a duality between the topological string partition functions computed on target spaces $C_{m,n}$ and $cA_{m+n-1}$. From the physical point of view, we are relating a topological string theory with two different target spaces. Note that the threefold $cA_{m+n-1}$ is actually a product $A_{m+n-1} \times \mathbb C$ 
of a singular surface and a line, thus two real dimensions decouple and we are left with a theory 
in 4 dimensions. In particular, the  partition function can be computed on a six dimensional generalised conifold $C_{m,n}$ or equivalently on the four dimensional $A_{m+n-1}$ orbifold (in section \ref{FR}, we give an explanation for this decoupling from the perspective of Gromov-Witten theory on $A_n \times \mathbb C$). Similar to what happens for T-duality in string theory we obtain two dual geometric target spaces. From the mathematical point view, we expect that there is an intimate relation between some topological invariants (for example, Gromov--Witten, Gopakumar--Vafa, Donaldson--Thomas invariants) defined on these two duals conifolds.

\section{Partition function on singular Calabi--Yaus} 
Given  the topological string partition functions for  theories with target  on smooth  varieties, the corresponding partition functions for theories with target on singular varieties were defined in \cite{GKMR}. 
Their method was to take an average over  partition functions calculated over all crepant resolutions. 
Although the same method applies whenever the singular space has a finite number of resolutions, 
their focus was on Calabi--Yau threefolds. We first recall their construction. 

\subsection{Toric Calabi--Yau threefolds}

Suppose we have a singular Calabi--Yau space $X$ and a finite
collection of crepant resolutions $X^t \to X$ for index $t \in
\mathcal{T}$, $\abs{\mathcal{T}} < \infty$. Assume further that
a partition function $Z_\text{old}(Y; Q, \dotsc)$ defined
for  smooth Calabi--Yau varieties $Y$ is known, where $Q = (Q_1, Q_2, \dotsc)$
are formal variables corresponding to a basis of $H_2(Y;\mathbb{Z})$.
Additionally, we suppose that $H_2(X^{s};\mathbb{Z}) \cong H_2(X^{t};\mathbb{Z})$
for all $s,t \in \mathcal{T}$.
That is, since, for the varieties we consider, all crepant resolutions have isomorphic second homologies, we fix an identification of their basis elements according to formula (\ref{idenQs}). This is simply a formal linear identification of the generators of the second homology groups, which does not require considerations of change of variables. It is worth noticing that the identification we use here is different from the change of variables often used in the literature when considering the basic flop, which inverts the Q variables, such as in [Sz]. As (\ref{idenQs}) shows, we do not invert any of the variables. Our choice just identifies the basis of second homology groups of crepant resolutions formally for the purpose of analysing the resulting new partition functions.\\

We define a new partition function $Z_\text{new}$ for $X$ as follows.
Firstly, we identify the formal variables $Q$ among all the
resolutions. That is, we identify
\be \label{idenQs}
Q_i^s = Q_i^t \ec Q_i \text{ \ for all \ } s,t \in \mathcal{T} \text{ .} 
\ee
More explicitly, for each resolution we label the edges of the corresponding triangulation so that $Q_i$ denotes the $i$-th interior edge when reading from left to right. Then, the isomorphism between $H_2(X^{s};\mathbb{Z})$ and $H_2(X^{t};\mathbb{Z})$ sends the class of the $i$-th interior edge in the
triangulation depicting $X^{s}$ to the class of the $i$-th interior edge in the triangulation depicting $X^{t}$ (for example, in figure \ref{resolutions-C23} each triangulation has four internal edges, drawn in red or green, which are labelled $Q_1, Q_2, Q_3, Q_4$ from left to right). Secondly, we define
\be \label{define}
Z_\text{new}(X;Q,\dotsc) \ce \prod_{t \in \mathcal{T}} Z_\text{old}(X^t; Q^t, \dotsc) \ ,
\ee
where $\dotsc$ stands for other possible variables that the partition function could depend on. The new partition function captures  information from all possible
resolutions of $X$, and thus can be regarded as describing the singular
variety $X$ itself.\\

The main result of \cite{GKMR} shows that the new partition function is 
homogeneous (see definition \ref{homogeneous} below). We will   show that this result implies that cones on weighted projective spaces (or more precisely, on abelian two dimensional orbifolds) have a privileged role amongst Calabi--Yau threefolds, in a sense that we will make precise using the topological string partition functions. \\

Let us consider the case of  toric Calabi--Yau threefolds without 
compact 4-cycles. Such a   threefold   $X$ corresponds  to a chain of $\mathbb {P}^1$'s. Over these $\mathbb {P}^1$'s we have the line bundles  $\mathcal O(-1) \oplus \mathcal O(-1) \to \mathbb {P}^1$ or $\mathcal O(-2) \oplus \mathcal O \to \mathbb {P}^1 $. The topological string partition function in this case is given by \cite{AKMV, IKP}
\be
Z_{\text{top}}(-q,Q)=(M(q))^{\chi} \; Z_\text{PT}(-q,Q) \ ,
\ee
where $M(q)$ is the MacMahon function
\be\label{MacMahon}
M(q):=\prod_{k=1} ^{\infty}\frac{1}{(1-q^k)^k} \ ,
\ee
and  the Pandharipande--Thomas partition function  is  (see section 2.4 of \cite{OSY})
\be
Z_{\text{PT}} (-q,Q) =\prod_{1\leq i\leq j\leq \chi-1} \prod_{k=1} ^{\infty}(1-Q_iQ_{i+1}\cdots Q_jq^k)^{-ks_is_{i+1}\cdots s_j} \  .
\ee
Moreover, $\chi$  is the Euler characteristics of $X$, $Q_1, \dots,Q_{\chi-1}$ are related to the
the K\"ahler parameters (we have $Q_i=\exp(-t^i)$ where $t^i$ are the K\"ahler parameters) associated to the $\mathbb P^1$'s, and from the physical perspective $q$ is related to the topological string coupling constant. Note that  the number of $\mathbb {P}^1$'s is $\chi-1$. There are two possible values for $s_i$, that is, $s_i = -1$ or $+1$ depending on whether the $i$-th $\mathbb {P}^1$ is resolved by $\mathcal O(-1) \oplus \mathcal O(-1) \to \mathbb {P}^1$ or $\mathcal O(-2) \oplus \mathcal O \to \mathbb {P}^1 $.
Note that  our identification of variables $Q_i$ in (\ref{idenQs}) results in also  identifying  the K\"ahler parameters $t_i$, thus giving a  different result in comparison to   the usual identification done using the flop, which changes the signs of the $t_i$'s.
\begin{dfn}\label{homogeneous}
A partition function $Z(q,Q)$ of variables $Q=(Q_1,Q_2,\cdots )$ is called \emph{homogeneous of degree $d$} if  it
 has the form
\be
Z(-q,Q)=\left( \prod_{1\leq i\leq j\leq \chi-1} \prod_{k=1} ^{\infty}(1-Q_iQ_{i+1}\cdots Q_jq^k)^{-k} \right)^d .
\ee
\end{dfn}

\begin{obs}
The definition of homogeneous partition function given in [GKMR] is slightly different from  ours. Here we 
have made a correction, namely to remove the MacMahon factor \eqref{MacMahon} to the power $\chi$  that appears in $Z$ but not in $Z_{\text{PT}} $. Such factor appears in each expression of Z for a given singular threefold as many times as the number of its crepant resolutions. Their main result  states that the new  partition function they define is homogeneous, 
but the precise statement actually requires our definition \ref{homogeneous}.
\end{obs}

\subsection{Generalised conifolds}

Given a pair of non-negative integers $m,n$, not both zero, we consider the toric varieties 
\be 
C_{m,n}:=\{(x,y,z,w)\ | \ xy-z^mw^n=0\}\subset \mathbb{C}^4=\text{Spec} \, \mathbb{C}[x,y,z,w]. 
\ee
We refer to this specific type of toric varieties as generalised conifolds.
\begin{dfn}\label{curvecounting}
A partition function for a Calabi--Yau manifold $Y$ is of \emph{curve-counting type} if it can be expressed in terms of the Donaldson--Thomas, Gromov--Witten or Gopakumar--Vafa partition function up to a factor depending only on the Euler characteristic of $Y$.
\end{dfn}


\begin{thm}\cite{GKMR}
\label{thm1}
Let $X$ be a toric Calabi--Yau threefold defined as a subset of $\mathbb{C}^4$ by $X=\mathbb{C}[x,y,z,w]/\langle xy-z^mw^n\rangle$, where $m$ and $n$ are integers not both zero. Let $Z(X^t;q,Q)$ be a partition function of curve-counting type. Then the total partition function 
\be
  Z_{\text{tot}}(X;q,Q):=\displaystyle\prod_{t \in \mathcal T}Z(X^t;q,Q), 
\ee
is homogeneous of degree
\begin{equation}
 d= \frac{(m+n-2)!}{m!n!}(m^2-m+n^2-n-2mn) \ ,
\end{equation}
when $m+n\geq 2$, otherwise $d=0$. 
\end{thm}

See \cite{GKMR} for a proof of this theorem and further details.\\

Now consider the threefold $C_{0,n}\cong cA_{n-1}$ given by $xy-w^n=0$. Section 4 of \cite{GKMR} observes that $C_{0,n}$ are quotients of $\mathbb{C}^3$ by $\mathbb{Z}/n\mathbb{Z}$ acting on a two-dimensional subspace $\mathbb{C}^2$ as $(a,b,c)\mapsto (\varepsilon a, \varepsilon^{-1}b, c)$, with $\varepsilon^n=1$. These spaces have $1$-dimensional singularities, as $C_{0,n}\cong A_{n-1} \times \mathbb{C}$, where $A_{n-1=}\{(x,y,z)\ | \ xy-w^n=0\}\subset \mathbb{C}^3$ is the $A_{n-1}$ orbifold (with a singular point at the origin). They play an important role amongst all the $C_{m,n}$ threefolds.

\begin{cor}\label{cor1}
Let $X$ be a toric Calabi--Yau threefold defined as a subset of $\mathbb{C}^4$ by $X=\mathbb{C}[x,y,z,w]/\langle xy-z^mw^n\rangle$, where $m$ and $n$ are non-negative integers. Let $Z( X ;q,Q)$ be a partition function of curve-counting type. Then the total partition function is given by 
\be
Z_\text{tot}(X;-q,Q)=\left( \left( M(q)\right)^{\chi}\right)^{m+n\choose n}\prod_{t \in \mathcal T}Z_{\text{PT}} (X^t;-q,Q).
\ee
where  $M(q)$ is the MacMahon function and $Z_{\text{PT}} (X^t;q,Q)$ represents the Pandharipande--Thomas partition function of the crepant resolution $t$.
\end{cor}
\begin{proof}
By  definition
\be
Z_{\textnormal{tot}}(X;q,Q)=\prod_{t \in \mathcal T}Z(X^t;q,Q).
\ee
We know, by \cite{GKMR}, that the number of crepant resolutions of $X$ is $m+n \choose n$.\\

Given $X^t$ a crepant resolution of $X$, it follows that
\be
Z(X^t;-q,Q)=\left( \prod_{k=1} ^\infty \frac{1}{(1-q^k)^k} \right)^{\chi} \prod_{1\leq i\leq j\leq \chi -1} \prod_{k=1} ^\infty (1-Q_iQ_{i+1}\cdots Q_jq^k)^{-ks_is_{i+1}\cdots s_j}.
\ee
So, when we run the product over  all  crepant resolutions  the term $\left( \prod_{k=1} ^\infty \frac{1}{(1-q^k)^k}\right)^{\chi}=(M(q))^{\chi}$ appears $m+n \choose n$ times. Therefore,
\be
Z_{\textnormal{tot}}(X;-q,Q)=\left( (M(q))^{\chi}\right)^{m+n\choose n}\prod_{t \in \mathcal T}Z_{\text{PT}} (X^t;-q,Q).
\ee
\end{proof}

\subsubsection{Partition functions on $C_{m,n}$ and $cA_{m+n-1}$}

Note that by Proposition 4.1 (1) of \cite{GKMR} the number of triangles on each triangulation of $C_{m,n}$ (which is $m+n$) is equal to the number of triangles in the unique triangulation of $C_{0,m+n}\cong cA_{m+n-1}$. It follows that the factors of the product
\be
\prod_{1 \leq i \leq j \leq \chi-1}(1-Q_i \cdots Q_jq^k)^{-ks_i\cdots s_j}
\ee
are equal up to the powers $s_k$'s; in fact, all such  resolutions have the same number of $Q_i'$s, hence  the summation runs over the same indices $1 \leq i \leq j \leq \chi-1$, only the exponents $s_i$ 
vary according to  the resolution. \\

Now let $X=\mathbb{C}[x,y,z,w]/\langle xy-z^mw^n\rangle$ and $Y=\mathbb{C}[x,y,z,w]/\langle xy-w^{m+n}\rangle$. The partition function of $C_{0,l}\cong cA_{l-1}$ has degree $d=1$ for all positive integer $l$, see Theorem \ref{thm1}. This implies that the partition function of $Y$ is homogeneous of degree 1, so we can write
\be
Z_{\text{PT}} (Y;-q,Q)=\prod_{1 \leq i \leq j \leq \chi - 1} \prod_{k=1}^\infty(1-Q_i\cdots Q_jq^k)^{-k}.
\ee
On the other hand, from Theorem \ref{thm1}, the partition function of $X$ is homogeneous of degree $d$. Therefore,
\be
Z_{\text{PT}} (X;-q,Q)=\left(\prod_{1 \leq i \leq j \leq \chi - 1} \prod_{k=1}^\infty(1-Q_i\cdots Q_jq^k)^{-k}\right)^d.
\ee
So we proved the following corollary of Theorem \ref{thm1}:
\begin{thm}\label{thm2}
Let $X$, $Y$ be singular Calabi--Yau threefolds defined as subsets of $\mathbb{C}^4$ by $X=\mathbb{C}[x,y,z,w]/\langle xy-z^mw^n\rangle$ and $Y=\mathbb{C}[x,y,z,w]/\langle xy-w^{m+n}\rangle$, respectively. Then
\be
Z_{\text{PT}} (X;q,Q)=Z_{\text{PT}} (Y;q,Q)^d,
\ee
where $d$ is the degree of $Z_{\textnormal{tot}}(X;q,Q)$.
\end{thm}

By Corollary \ref{cor1} and Corollary \ref{thm2} the partition function of $X=\mathbb{C}[x,y,z,w]/\langle xy-z^mw^n \rangle$ is completely described by
\begin{equation} \label{expressao}
Z_{\textnormal{tot}}(X;-q,Q)=\left( (M(q))^{\chi}\right)^{m+n\choose m} Z_{\text{PT}} (Y;-q,Q)^d,
\end{equation}
where $Y=\mathbb{C}[x,y,z,w]/\langle xy-w^{m+n}\rangle$. \\

Thus, we have explicitly shown a relation between the topological string partition function on a Calabi--Yau threefold of the form $X=\mathbb{C}[x,y,z,w]/\langle xy-z^mw^n\rangle$, hence a generalised conifold $C_{m,n}$, and one of the form $Y=\mathbb{C}[x,y,z,w]/\langle xy-w^{m+n}\rangle$, a compound du Val singularity $cA_{m+n-1}$. As we have pointed out the computation on $C_{0,m+n} \cong cA_{m+n-1}$ is simpler because $C_{m,n}$ has ${m+n \choose m}$ resolutions whereas  $cA_{m+n-1}$ has only one. \\

\section{Example}
We now illustrate the different ways to calculate partition functions exploring  one concrete example in full details. 
Consider $X=C_{2,3} :=\mathbb{C} [x,y,z,w]/\langle xy-z^2w^3\rangle$, so that the crepant resolutions of $X$ are represented by
$X_1, \dots, X_{10}$ as depicted on Figure \ref{resolutions-C23}, 
where a green line represents a resolution by $\mathcal O(-1) \oplus \mathcal O(-1) \to \mathbb {P}^1$ and a red line represents a resolution by $\mathcal O(-2) \oplus \mathcal O \to \mathbb {P}^1 $.
We can  compute the partition function of the singular threefold $X$ in two ways: either directly from definition  (\ref{define}), or through  expression \eqref{expressao} that uses the threefold $C_{0,5}$. \\

Firstly, to calculate directly from the definition, we  write down  the partition function for each crepant resolution. The 10 crepant resolutions  and their corresponding Pandharipande--Thomas partition functions are:

\begin{figure}\label{draw}
\begin{center}
\includegraphics[width=0.8\textwidth]{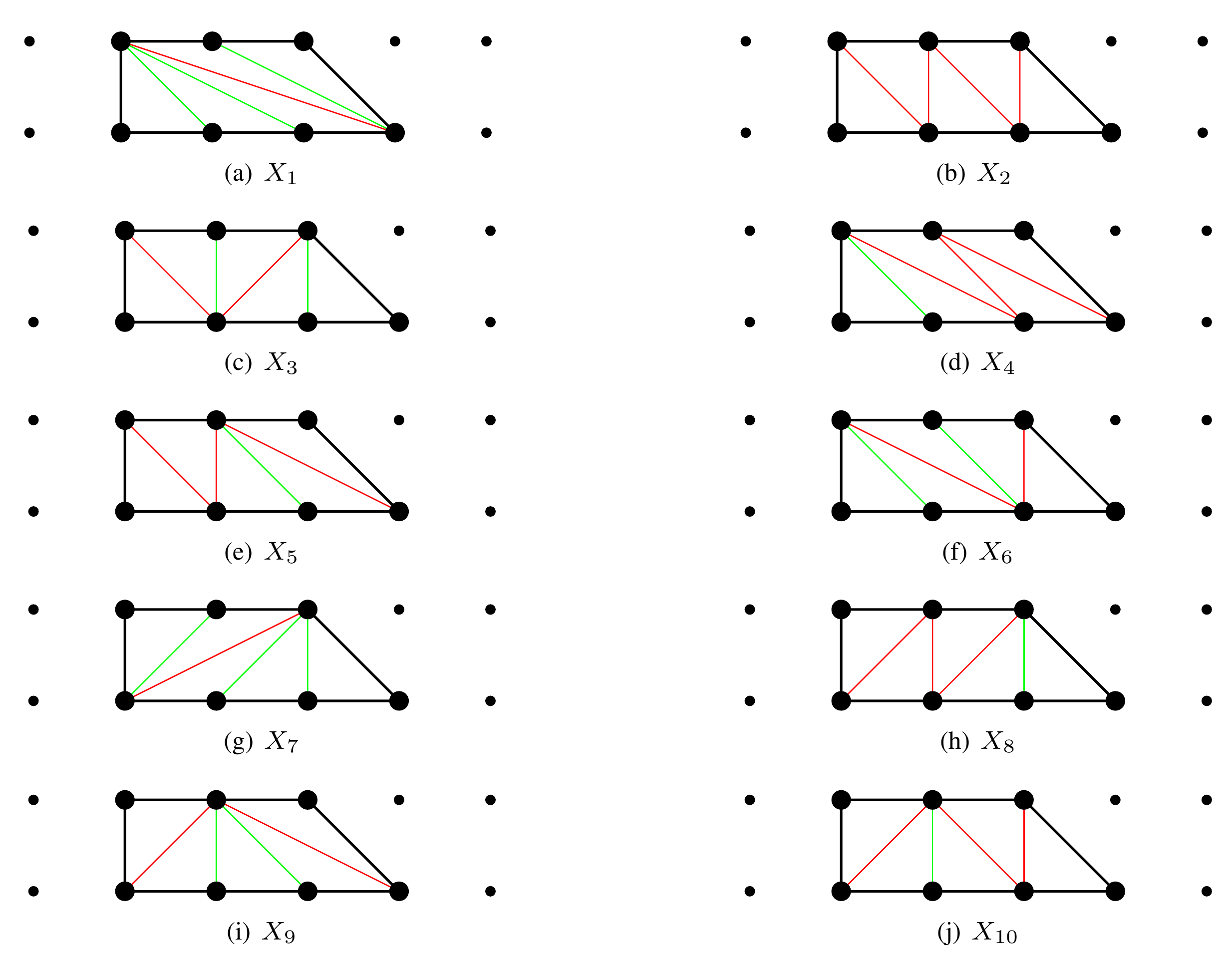}
\caption{Crepant resolutions of $C_{2,3}$}
\label{resolutions-C23}
\end{center}
\end{figure}

\begin{eqnarray*}
Z_\text{PT}(X_1,-q,Q)&=&\prod_{k=1} ^{\infty} (1-Q_1q^k)^{-k} (1-Q_1Q_2q^k)^{-k} (1-Q_1Q_2Q_3q^k)^{k} (1-Q_1Q_2Q_3Q_4q^k)^{k} \\
& & (1-Q_2q^k)^{-k} (1-Q_2Q_3q^k)^{k} (1-Q_2Q_3Q_4q^k)^{k} (1-Q_3q^k)^{k} (1-Q_3Q_4q^k)^{k} (1-Q_4q^k)^{-k}
\end{eqnarray*}

\begin{eqnarray*}
Z_\text{PT}(X_2,-q,Q)&=&\prod_{k=1} ^{\infty} (1-Q_1q^k)^{k} (1-Q_1Q_2q^k)^{-k} (1-Q_1Q_2Q_3q^k)^{k} (1-Q_1Q_2Q_3Q_4q^k)^{-k} \\
& & (1-Q_2q^k)^{k} (1-Q_2Q_3q^k)^{-k} (1-Q_2Q_3Q_4q^k)^{k} (1-Q_3q^k)^{k} (1-Q_3Q_4q^k)^{-k} (1-Q_4q^k)^{k}
\end{eqnarray*}

\begin{eqnarray*}
Z_\text{PT}(X_3,-q,Q)&=&\prod_{k=1} ^{\infty} (1-Q_1q^k)^{k} (1-Q_1Q_2q^k)^{k} (1-Q_1Q_2Q_3q^k)^{-k} (1-Q_1Q_2Q_3Q_4q^k)^{-k} \\
& & (1-Q_2q^k)^{-k} (1-Q_2Q_3q^k)^{k} (1-Q_2Q_3Q_4q^k)^{k} (1-Q_3q^k)^{k} (1-Q_3Q_4q^k)^{k} (1-Q_4q^k)^{-k}
\end{eqnarray*}

\begin{eqnarray*}
Z_\text{PT}(X_4,-q,Q)&=&\prod_{k=1} ^{\infty} (1-Q_1q^k)^{-k} (1-Q_1Q_2q^k)^{k} (1-Q_1Q_2Q_3q^k)^{-k} (1-Q_1Q_2Q_3Q_4q^k)^{k} \\
& & (1-Q_2q^k)^{k} (1-Q_2Q_3q^k)^{-k} (1-Q_2Q_3Q_4q^k)^{k} (1-Q_3q^k)^{k} (1-Q_3Q_4q^k)^{-k} (1-Q_4q^k)^{k}
\end{eqnarray*}

\begin{eqnarray*}
Z_\text{PT}(X_5,-q,Q)&=&\prod_{k=1} ^{\infty} (1-Q_1q^k)^{k} (1-Q_1Q_2q^k)^{-k} (1-Q_1Q_2Q_3q^k)^{-k} (1-Q_1Q_2Q_3Q_4q^k)^{k} \\
& & (1-Q_2q^k)^{k} (1-Q_2Q_3q^k)^{k} (1-Q_2Q_3Q_4q^k)^{-k} (1-Q_3q^k)^{-k} (1-Q_3Q_4q^k)^{k} (1-Q_4q^k)^{k}
\end{eqnarray*}

\begin{eqnarray*}
Z_\text{PT}(X_6,-q,Q)&=&\prod_{k=1} ^{\infty} (1-Q_1q^k)^{-k} (1-Q_1Q_2q^k)^{k} (1-Q_1Q_2Q_3q^k)^{k} (1-Q_1Q_2Q_3Q_4q^k)^{-k} \\
& & (1-Q_2q^k)^{k} (1-Q_2Q_3q^k)^{k} (1-Q_2Q_3Q_4q^k)^{-k} (1-Q_3q^k)^{-k} (1-Q_3Q_4q^k)^{k} (1-Q_4q^k)^{k}
\end{eqnarray*}

\begin{eqnarray*}
Z_\text{PT}(X_7,-q,Q)&=&\prod_{k=1} ^{\infty} (1-Q_1q^k)^{-k} (1-Q_1Q_2q^k)^{k} (1-Q_1Q_2Q_3q^k)^{k} (1-Q_1Q_2Q_3Q_4q^k)^{k} \\
& & (1-Q_2q^k)^{k} (1-Q_2Q_3q^k)^{k} (1-Q_2Q_3Q_4q^k)^{k} (1-Q_3q^k)^{-k} (1-Q_3Q_4q^k)^{-k} (1-Q_4q^k)^{-k}
\end{eqnarray*}

\begin{eqnarray*}
Z_\text{PT}(X_8,-q,Q)&=&\prod_{k=1} ^{\infty} (1-Q_1q^k)^{k} (1-Q_1Q_2q^k)^{-k} (1-Q_1Q_2Q_3q^k)^{k} (1-Q_1Q_2Q_3Q_4q^k)^{k} \\
& & (1-Q_2q^k)^{k} (1-Q_2Q_3q^k)^{-k} (1-Q_2Q_3Q_4q^k)^{-k} (1-Q_3q^k)^{k} (1-Q_3Q_4q^k)^{k} (1-Q_4q^k)^{-k}
\end{eqnarray*}

\begin{eqnarray*}
Z_\text{PT}(X_9,-q,Q)&=&\prod_{k=1} ^{\infty} (1-Q_1q^k)^{k} (1-Q_1Q_2q^k)^{k} (1-Q_1Q_2Q_3q^k)^{k} (1-Q_1Q_2Q_3Q_4q^k)^{-k} \\
& & (1-Q_2q^k)^{-k} (1-Q_2Q_3q^k)^{-k} (1-Q_2Q_3Q_4q^k)^{k} (1-Q_3q^k)^{-k} (1-Q_3Q_4q^k)^{k} (1-Q_4q^k)^{k}
\end{eqnarray*}

\begin{eqnarray*}
Z_\text{PT}(X_{10},-q,Q)&=&\prod_{k=1} ^{\infty} (1-Q_1q^k)^{k} (1-Q_1Q_2q^k)^{k} (1-Q_1Q_2Q_3q^k)^{-k} (1-Q_1Q_2Q_3Q_4q^k)^{k} \\
& & (1-Q_2q^k)^{-k} (1-Q_2Q_3q^k)^{k} (1-Q_2Q_3Q_4q^k)^{-k} (1-Q_3q^k)^{k} (1-Q_3Q_4q^k)^{-k} (1-Q_4q^k)^{k} \text{.}
\end{eqnarray*}

After calculating  these 10  partition functions  individually for the crepant resolutions of $X$ we take their product, thus obtaining:

\begin{eqnarray*} 
Z_\text{PT}(C_{2,3};-q,Q)&=&\prod_{k=1} ^{\infty} (1-Q_1q^k)^{2k} (1-Q_1Q_2q^k)^{2k} (1-Q_1Q_2Q_3q^k)^{2k} (1-Q_1Q_2Q_3Q_4q^k)^{2k} \\
& & (1-Q_2q^k)^{2k} (1-Q_2Q_3q^k)^{2k} (1-Q_2Q_3Q_4q^k)^{2k} (1-Q_3q^k)^{2k} (1-Q_3Q_4q^k)^{2k} (1-Q_4q^k)^{2k}\\
&=& \Bigg( \prod_{k=1} ^{\infty} (1-Q_1q^k)^{-k} (1-Q_1Q_2q^k)^{-k} (1-Q_1Q_2Q_3q^k)^{-k} (1-Q_1Q_2Q_3Q_4q^k)^{-k} \\
& & (1-Q_2q^k)^{-k} (1-Q_2Q_3q^k)^{-k} (1-Q_2Q_3Q_4q^k)^{-k} (1-Q_3q^k)^{-k} (1-Q_3Q_4q^k)^{-k} (1-Q_4q^k)^{-k}\Bigg)^{-2} 
\text{.}
\end{eqnarray*} 

Secondly, using  expression \eqref{expressao} we consider the auxiliary threefold  $Y = C_{0, 5}:=\mathbb{C}[x,y,z,w]/\langle xy-w^5\rangle$ 
which has a   unique crepant resolution, see Figure \ref{resolutions-C05}, and  partition function given by:

\begin{eqnarray*}
Z_\text{PT}(C_{0,5};-q,Q)&=&\prod_{k=1} ^{\infty}(1-Q_1q^k)^{-k} (1-Q_1Q_2q^k)^{-k} (1-Q_1Q_2Q_3q^k)^{-k} (1-Q_1Q_2Q_3Q_4q^k)^{-k} \\
& & (1-Q_2q^k)^{-k} (1-Q_2Q_3q^k)^{-k} (1-Q_2Q_3Q_4q^k)^{-k} (1-Q_3q^k)^{-k} (1-Q_3Q_4q^k)^{-k} (1-Q_4q^k)^{-k}.
\end{eqnarray*}

In this simple example is easy to verify Theorem \ref{thm2} which states that the Pandharipande--Thomas partition function of $C_{2,3}$ and $C_{0,5}$ are related by the following equation

\[
Z_{\text{PT}} (C_{2,3};q,Q)=Z_{\text{PT}} (C_{0,5};q,Q)^d,
\]

where $d=\frac{(m+n-2)!}{m!n!}(m^2-m+n^2-n-2mn)=-2$.

\begin{figure}
\begin{center}
\includegraphics[width=0.4\textwidth]{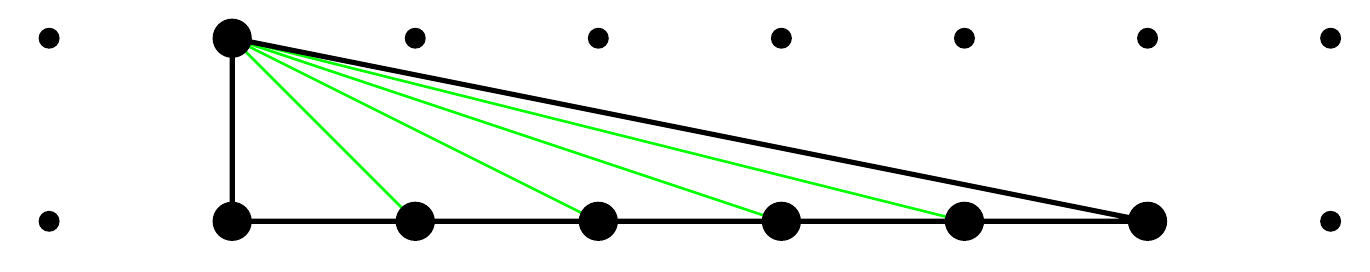}
\caption{Crepant resolution of $C_{0,5}$}
\label{resolutions-C05}
\end{center}
\end{figure}

\section{Final remarks}\label{FR}

Although the other compound du Val singularities $cD$ and $cE$ type are not toric varieties, it could be expected that 
in these cases there exists as well a dual variety in 4D with equivalent topological string partition function. We plan to address this issue in future work. \\

D. Maulik gave a complete solution for the reduced Gromov-Witten theory of $A_n$ singularities, for any genus and arbitrary descendent insertions  in \cite{Ma}. He also studied and described the threefold $A_n \times \mathbb P^1$. Now, if instead of $\mathbb P^1$ we have the complex plane $\mathbb C$, that is, we want to describe Gromov-Witten theory on $A_n \times \mathbb C$, then any curve gets contracted under the projection to $\mathbb C$, so it is basically the same as the Gromov-Witten theory of $A_n$ with a Hodge class inserted. The solution on the threefold $A_n \times \mathbb C$ can also be obtained from the $A_n \times \mathbb P^1$ studied in \cite{Ma} by restricting to the $\beta=0$ constant map case (we are grateful to Davesh Maulik for explaining this point to us).

Hopefully the proposed duality extends to arbitrary descendent insertions (or, from a physical point of view, equivalence of correlation functions of operators described by descendent insertions) and not only to the partition function. Then, using the correspondence between $C_{m,n}$ and $A_{m+n-1} \times \mathbb C$, in principle, we should be able to write the complete solution for Gromov-Witten theory on generalised conifolds $C_{m,n}$. In other words, the Gromov-Witten invariants on $C_{m,n}$ could be equivalently computed on the simpler variety $A_{m+n-1} \times \mathbb C$. We will tackle this point in future work.\\

As suggested by Maulik the Gromov-Witthen theory on  $A_n \times \mathbb C$ is basically the same as the Gromov-Witten theory of $A_n$ with a Hodge class inserted. Thus, we would like to emphasise that the duality that we have been describing is essentially between the six dimensional variety $C_{m,n}$ and the four dimensional one $A_{m+n-1}$, \\


\begin{acknowledgments}
The idea of comparing the partition functions on $C_{m,n}$ and $C_{0,m+n}$ came from a discussion of Piotr Su\l kowski and E.G. We are thankful to P. Su\l kowski whose suggestions inspired this work. We thank Davesh Maulik for clarifying the relation between Gromov-Witten theory on the varieties: $A_m \times \mathbb C$, $A_m$ and $A_m \times \mathbb P^1$.\\
This collaboration started during a visit of  E. Gasparim to the department of Physics of UNAB in Santiago. Our special thanks to Per Sundell for his hospitality  and support under  CONICYT  grant   DPI 20140115.\\
B. Suzuki would like to thank Conicyt for the financial support through Beca Doctorado Nacional - Folio  21160257. C. A. B. Varea was partially supported by the Vice Rector\'ia  de Investigaci\'on y Desarrollo Tecnol\'ogico of UCN, Chile. The work of A. Torres-Gomez is funded by Conicyt grant PAI/ACADEMIA 79160014. Moreover, A. Torres-Gomez was partially supported by the National Research Foundation of Korea through the grant NRF-2014R1A6A3A04056670 at the final stage of this work.\\
Finally, we are very grateful to the referee for the careful reading and valuable suggestions.
\end{acknowledgments}


\end{document}